\documentclass[runningheads]{llncs}
\usepackage[normalem]{ulem}
\usepackage[T1]{fontenc}
\usepackage{adjustbox}
\usepackage{multirow}
\usepackage{graphicx}
\usepackage{xcolor}
\usepackage{amsmath,amssymb,amsfonts}
\usepackage{mathpartir}
\usepackage{hyperref}

\spnewtheorem{observation}{Observation}{\bfseries}{\itshape}

\usepackage[linesnumbered, noend]{algorithm2e}
\SetKwComment{Comment}{/* }{ */}
\SetKwInOut{Input}{input}
\SetKwInOut{Output}{output}
\SetKwRepeat{Do}{do}{while}
\RestyleAlgo{ruled}
\usepackage{comment}

\newcommand{\BWT}{\ensuremath{\mathrm{BWT}}}
\newcommand{\SA}{\ensuremath{\mathrm{SA}}}

\newcommand{\LCP}{\ensuremath{\mathrm{LCP}}}
\newcommand{\BWTr}{\ensuremath{\mathrm{BWT}_r}}
\newcommand{\SAr}{\ensuremath{\mathrm{SA}_r}}
\newcommand{\LCPr}{\ensuremath{\mathrm{LCP}_r}}
\newcommand{\PSV}{\ensuremath{\mathrm{PSV}}}
\newcommand{\NSV}{\ensuremath{\mathrm{NSV}}}

\newcommand{\SSS}{\mathcal S}

\newcommand{\rev}[1]{#1^{\text{rev}}}

\begin{document}
\title{Testing Suffixient Sets}
%
%
\author{Davide Cenzato\inst{1}\thanks{Funded by the European Union (ERC, REGINDEX, 101039208). Views and opinions expressed are however those of the author(s) only and do not necessarily reflect those of the European Union or the European Research Council. Neither the European Union nor the granting authority can be held responsible for them.}\orcidID{0000-0002-0098-3620} \and
Francisco Olivares\inst{2,3}\thanks{Funded by Basal Funds FB0001 and AFB240001, and Fondecyt Grant 1260080, ANID, Chile.}\orcidID{0000-0001-7881-9794} \and
Nicola Prezza\inst{1}$^{\star}$\orcidID{0000-0003-3553-4953}}
\authorrunning{Davide Cenzato, Francisco Olivares and Nicola Prezza}
%
\institute{DAIS, Ca' Foscari University of Venice, Italy\\
\email{\{davide.cenzato,nicola.prezza\}@unive.it}
\and
Centre for Biotechnology and Bioengineering, Chile \and 
Department of Computer Science, University of Chile, Chile\\
\email{folivares@uchile.cl}
}
\maketitle              
\begin{abstract}
\emph{Suffixient sets} are a novel prefix array (PA) compression technique based on \emph{subsampling} PA (rather than compressing the entire array like previous techniques used to do): by storing very few entries of PA (in fact, a compressed number of entries), one can prove that pattern matching via binary search is still possible provided that random access is available on the text. In this paper, we tackle the problems of determining whether a given subset of text positions is (1) a suffixient set or (2) a suffixient set of minimum cardinality. We provide linear-time algorithms solving these problems.

\end{abstract}
\section{Introduction}



Suffixient sets (see Definition \ref{def:suffixient} for a formal definition) are a new suffix array (SA) \cite{ManberM93} compression technique introduced by Depuydt et al.\ in 2024 \cite{Depuydt24,TOCSpaper}. Unlike previous approaches, this method compresses the Prefix Array (PA), i.e., the mirror version of the SA storing the co-lexicographic order of the text's prefixes, by subsampling a small set of positions corresponding to right-maximal substrings in the original text, rather than encoding the full PA. In \cite{TOCSpaper}, we proved that this subset of PA positions is sufficient to support pattern matching queries \cite{Navarro21} via binary search, just as with the original suffix array, when providing random access to the text. In particular, this novel compression scheme outperforms the $r$-index \cite{GagieNP20} both in time and space due to a better memory locality and to the fact that the size of the smallest suffixient set is often smaller than the number of equal-letter runs in the Burrows-Wheeler transform (the parameter that ultimately determines the size of the r-index). 

The size of indexes based on suffixient sets, such as the suffixient array in \cite{TOCSpaper}, is directly influenced by the cardinality of these sets. This underlines the importance of computing a suffixient set with the fewest amount of samples. In a previous work presented at SPIRE \cite{CenzatoOP24}, we addressed the problem of computing a suffixient set of minimum cardinality for any given text. 
That work, however, left open the following complementary (and natural) questions:


\begin{problem}[Suffixiency]\label{problem1}
    Given set $S \subseteq [1,n]$ and a text $T \in \Sigma^n$, is $S$ a suffixient set for $T$?
\end{problem}

\begin{problem}[Minimum suffixiency]\label{problem2}
    Given a suffixient set $S \subseteq [1,n]$ and a text $T \in \Sigma^n$, is $S$ a suffixient set of minimum cardinality for $T$?
\end{problem}

Observe that these natural problems were tackled for string attractors in \cite{kempa2018string}. Suffixient sets are indeed string attractors \cite{TOCSpaper}, but possess additional properties which make the algorithm in \cite{kempa2018string} not suitable for solving the above problems. 


In this paper, we show that Problems \ref{problem1} and \ref{problem2} can be solved in linear time. An implementation of our algorithms is publicly available online in the repository {\tt \url{https://github.com/regindex/suffixient}}.

\section{Preliminaries}\label{sec:preliminaries}
Given $i, j, n \in \mathbb{N}$ such that $1, \leq i, j \leq n$, and given a finite alphabet $\Sigma = \{1, \dots, \sigma\}$, we use the following notation: $[i, j] 
= \{i, i + 1, \dots, j\}$, if $i \leq j$ and $[i, j] = \emptyset$, otherwise; $[n]$ for the interval $[1, n]$; $T = T[1, n]$ for a length-$n$ string $T$ over $\Sigma$; $\vert T \vert$ for the length of $T$; $\varepsilon$ for the empty string, which is the only string satisfying $\vert \varepsilon \vert = 0$; $T[i]$ for the $i$-th character of $T$; $T[i, j]$ for the substring $T[i] \dots T[j]$, if $i \leq j$, and $T[i, j] = \varepsilon$, otherwise; $T[1, i]$ for a prefix of $T$ ending a position $i$, and $T[j, n]$ for a suffix of $T$ starting at position $j$. We define $\rev{T} = T[n] \dots T[1]$. 

Given an integer array $\mathcal{A}$, 
the \textit{previous smaller value} array of $\mathcal{A}$ is an integer array of length $\vert \mathcal{A} \vert$ defined as $\PSV(\mathcal{A})[i] = \max(\{j \,|\, j < i, A[j] < A[i]\}\cup \{0\})$ for all $i\in [\vert \mathcal{A} \vert]$. The \textit{next smaller value} array of $\mathcal{A}$ is defined as $\NSV(\mathcal{A})[i] = \min(\{j \,|\, j > i,A[j] < A[i]\}\cup \{ \vert \mathcal{A} \vert  + 1 \})$ for all $i\in [\vert \mathcal{A} \vert]$.

We assume the reader to be familiar with the concepts of suffix array ($\SA$), Burrows-Wheeler Transform ($\BWT$), and Longest Common Prefix array ($\LCP$). For space reasons, we give the formal definitions in Appendix~\ref{sec:deferred-preliminaries}. By $\BWTr$, $\SAr$ and $\LCPr$ we mean $\BWT(\rev{T})$, $\SA(\rev{T})$, and $\LCP(\rev{T})$, respectively. Given $T \in \Sigma^n$ and $i \in [n]$, by $text(i) = n - \SAr[i] + 1$ we indicate the position in $T$ corresponding to the character $\BWTr[i]$. Moreover, given $x \in [n]$ we denote by $bwt(x) = \SAr^{-1}[n + 1 - x]$ the position in $\BWTr$ of the character $T[x]$. 
For any $bwt(x) = k$, we define 
$bwt(x)_{max} = k_{max} = k + 1$ if $k < \vert \BWTr \vert$ and $\LCPr[k] \leq \LCPr[k + 1]$ 
and $bwt(x)_{max} = k_{max} = k$, otherwise. 
For any $i \in [2, n]$ and $a, c \in \Sigma$, we say $i$ is a $c$-run break if $\BWTr[i- 1, i] = ac$ or $\BWTr[i - 1, i] = ca$ and $a \neq c$. We use $box(k)$ to denote the maximal interval $[l_k, r_k]$ such that $k \in [l_k, r_k]$ and $\LCPr[k] \leq \LCPr[h]$, for $h \in [l_k, r_k]$. We compute $box(k) = [l_k, r_k] = [\PSV[k] + 1, \NSV[k] - 1]$ in $O(1)$ time, where by $\PSV$ and $\NSV$ we mean $\PSV(\LCPr)$ and $\NSV(\LCPr)$, respectively.


We say that $T[x, y]$ ($x - 1 \leq y$) is a right-maximal substring of $T \in \Sigma^n$ if $T[x, y]$ is a suffix of $T$ or there exists $a, b \in \Sigma$ such that $a \neq b$ and both $T[x, y]\cdot a$ and $T[x, y]\cdot b$ occur in $T$. Also, for $c \in \{a, b\}$ we say that $T[x, y]\cdot c$ is a one-character right-extension. In particular, we say that $T[x, y]\cdot c$ is a $c$-right-extension. Note that the empty string $\varepsilon$ is also right-maximal.

\subsection{Suffixient sets}
We start this section by recalling the suffixient set definition.

\begin{definition}[Suffixient set \cite{CenzatoOP24,Depuydt24}]\label{def:suffixient}
A set $S \subseteq [n]$ is  \emph{suffixient} for a string $T$ if, for every one-character right-extension $T[i,j]$ ($j\geq i$) of every right-maximal string $T[i,j-1]$, there exists $x\in S$ such that  $T[i,j]$ is a suffix of $T[1,x]$.
\end{definition}


In \cite{CenzatoOP24}, Cenzato et al. studied methods to build suffixient sets of smallest cardinality. All of their algorithms that compute a smallest suffixient set rely on the notion of \textit{supermaximal extensions}.

\begin{definition}[Supermaximal extensions \cite{CenzatoOP24,TOCSpaper}]\label{def:supermaximal-extensions}
We say that $T[i,j]$ (with $j \ge i$) is a \emph{supermaximal extension} if $T[i,j-1]$ is right-maximal, and for each right-maximal $T[i',j'-1] \neq T[i,j-1]$ (with $i' \le j' \le n$), $T[i,j]$ is not a suffix of $T[i',j']$.
\end{definition}

As proven in \cite{CenzatoOP24,TOCSpaper}, the set described in the following definition is, indeed, a suffixient set of the smallest cardinality.

\begin{definition}[Suffixient set of smallest cardinality, \cite{CenzatoOP24,TOCSpaper}]\label{def:smallest-suffixient}
Let $<_t$ be any total order on $[n]$. We define a set $\SSS \subseteq [n]$ as follows: $x \in \SSS$ if and only if there exists a supermaximal extension $T[i,j]$ such that (i) $T[i,j]$ is a suffix of $T[1,x]$, and (ii) for all prefixes $T[1,y]$ suffixed by $T[i,j]$, if $y\neq x$ then $y <_t x$.
\end{definition}

\section{Testing suffixient sets}\label{sec:testing-suffixiency}
In the search for a method to solve Problem~\ref{problem1}, as a first step, by Definition~\ref{def:suffixient} it seems natural to find a strategy to capture all of the right-maximal substrings of $T$. In that sense, the following observation will be useful.

\begin{observation}\label{obs:run-break-right-ext}
For each $c$-right-extension $T[k, l]$ $(k \leq l)$ of every right-maximal string $T[k, l - 1]$, there exists at least one $c$-run break $i$ such that $T[k, l]$ is a suffix of $T[text(i') - \LCPr[i], text(i')]$, with $i' \in \{i - 1, i\}$ and $\BWTr[i'] = c$.
\end{observation}
This observation induces an operative method to solve Problem~\ref{problem1}: for each run-break $i$, with $i' \in \{i - 1, i\}$ evaluate if there exists $x \in \SSS$ such that $T[text(i') - \LCPr[i], text(i')]$ suffixes $T[1, x]$.

In the following, 
in Lemma~\ref{lem:rm-suffixes-i_max} we show a method to determine if one string suffixes another in $O(1)$ time, so we can easily develop a quadratic algorithm solving Problem~\ref{problem1} (Section~\ref{sec:quadratic}), and in Section~\ref{sec:linear} we present the linear-time version of this method. Finally, in Section~\ref{sec:cardinality} we present a method to solve Problem~\ref{problem2} by generalizing the approach used in the previous sections.

\subsection{A simple quadratic algorithm for Suffixiency}\label{sec:quadratic}

From the following lemma, we obtain a method to evaluate if $T[1, x]$ is suffixed by a right-extension $T[y, z]$ in $O(1)$ time. 

\begin{lemma}\label{lem:rm-suffixes-i_max}
Given $x \in [n]$, let $k = bwt(x)$, and let $j$ be a $c$-run break, such that $j' \in \{j - 1, j\}$ and $\BWTr[j'] = c$. We have, $P = T[text(j') - \LCPr[j], text(j')]$ suffixes $T[1, x]$ if and only if $\BWTr[k] = c$ and $\PSV[j] < k_{\max} < \NSV[j]$.
\end{lemma}
\begin{proof}
\noindent($\Rightarrow$) Let $Q = T[x - \LCPr[k_{\max}], x]$ and assume that $P$ suffixes $T[1, x]$. By definition of $\SAr$, $\LCPr[k_{max}]$ is the length of the largest suffix that $T[1, x]$ shares with any other substring of $T$, then $P$ also suffixes $Q$. In particular, $P' = P[1, \vert P \vert - 1]$ suffixes $Q' = Q[1, \vert Q \vert - 1]$, so the area in $\SAr$ corresponding to $\rev{(Q')}$ is included in the area in $\SAr$ of $\rev{(P')}$, which implies $\PSV[j] < k_{\max} < \NSV[j]$. In addition, since $P$ suffixes $T[1, x]$, we have $T[text(j')] = T[x] = \BWTr[j'] = \BWTr[k] = c$.

\noindent($\Leftarrow$) We proceed by the contrapositive. Assume $P$ does not suffixes $T[1, x]$. Then, we have $T[text(j')] = \BWTr[j'] \neq \BWTr[k] = T[x]$ or $k_{\max} \not\in box(j)$, which implies $k_{\max} \leq \PSV[j]$ or $\PSV[j] \leq k_{\max}$. \qed
\end{proof}


From Observation~\ref{obs:run-break-right-ext} and Lemma~\ref{lem:rm-suffixes-i_max} we immediately obtain a simple algorithm to solve Problem~\ref{problem1}: for each $c$-run break $i$ in $\BWTr$ evaluate if there exists $x \in \SSS$ such that $bwt(x)_{\max} \in box(i)$ and $\BWTr[bwt(x)] = c$. Since there are $O(n)$ run-breaks in $\BWTr$, and in the worst case we have to scan the whole set $\SSS$ to find such an $x$, and $\vert \SSS \vert = O(n)$, this simple algorithm runs in $O(n^2)$ time.

\subsection{A linear algorithm for Suffixiency}\label{sec:linear}

In this section, we propose a method to evaluate each position $x$ in $\SSS$ only while $T[1, x]$ is suffixed by the strings induced by the run-breaks we process, ultimately leading to a linear-time algorithm.

In a first step, we take advantage of the following.
\begin{lemma}\label{lem:not-in-box}
Let $i \in [n]$ and let $i \neq j \in [n]$ such that $j \not\in box(i)$. If $j < i$ then for any $k < j$, we have $k \not\in box(i)$. On the other hand, if $i < j$, then for any $j < g$, it holds $g \not \in box(i)$. 
\end{lemma}
\begin{proof}
Since $j \not\in box(i)$, we have $j \leq \PSV[i]$ or $\NSV[i] \leq j$. In the former case, we have $j < i$, so $ k < j \leq \PSV[i]$ holds for any $k < j$, then $k \not\in box(i)$.
In the second case, we have $i < j$, thus $\NSV[i] < j < g$ holds for any $j < g$, so $g \not\in box(i)$.\qed
\end{proof}

Suppose that for some $c$-run break $i$ (with $i' \in \{i - 1, i\}$ and $\BWTr[i'] = c$) we want to determine if there exists $x \in \SSS$ such that $P = T[text(i') - \LCPr[i], text(i')]$ suffixes $T[1, x]$. Lemma~\ref{lem:not-in-box} tells us that if $bwt(x)_{\max} \not\in box(i)$, we have two cases: (i) if $bwt(x)_{\max} < i$, then for any $y \in \SSS$ such that $bwt(y)_{\max} < bwt(x)_{\max}$ we have $P$ also does not suffix $T[1, y]$; otherwise, (ii) for any $z \in \SSS$ such that $bwt(x)_{\max} < bwt(z)_{\max}$ we have $P$ also does not suffix $T[1, z]$.

Following the above idea, we propose the first linear-time algorithm to solve Problem \ref{problem1}. This procedure is summarized in Algorithm~\ref{alg:linear} in Appendix~\ref{sec:algs}.

The algorithm maintains three sets during the computation: a set $\mathcal{A} = \{bwt(x) \mid x \in \SSS\}$ (i.e., the elements of $\SSS$ replaced by their position in $\BWTr$), a set $\mathcal{B} = sort(\mathcal{A})$, and a set $\mathcal{C}$ such that $\mathcal{C}[c][k] = h_{\max}$, where $h$ is the $k$-th entry in $\mathcal{B}$ with $\BWTr[h] = c$. Now, for a given $c$-run break $i$, suppose that we scan $\mathcal{C}[c]$ from $1$ to $a \in [n]$ until we find for the first entry such that $\mathcal{C}[c][a] \in box(i)$, then for the next $c$-run break $j$ (with $j' \in \{j - 1, j\}$ and $\BWT[j'] = c$) there is no need to scan positions $1, \dots, a - 1$ of $\mathcal{C}[c]$, since we know that $T[text(j') - \LCPr[j], text(j')]$ cannot suffix any of the prefixes of $T$ induced by those positions. Using this observation, we keep, for each $c \in \Sigma$, a pointer in $\mathcal{P}[c]$, pointing to the last visited element in $\mathcal{C}[c]$. Then, each time we find a $c$-run break $i$, we start comparing from the last pointed position of $\mathcal{C}[c]$. If we reach a position $h > i$, and $h_{\max} \not\in box(i)$, then we know that $\SSS$ cannot be suffixient. 

Algorithm~\ref{alg:linear} runs in linear time and space. Arrays $\BWTr$, $\SA$, $\LCPr$, $\PSV$, and $\NSV$ consume $O(n)$ space. Arrays $\mathcal{A}$, $\mathcal{B}$, and $\mathcal{C}$ use $\vert \mathcal{S} \vert = O(n)$, and the array 
$\mathcal{P}$ uses $O(\sigma)$ space. So, the total algorithm space consumption is $O(n)$. Computing arrays on lines~\ref{line:linear_computing-arrays} and \ref{line:linear_computing-psv-nsv} takes $O(n)$ time \cite{BerkmanSV93,ManberM93}. In addition, computing the sets $\mathcal{A}$ (by computing $\SAr^{-1}[n - s + 1]$), 
$\mathcal{B}$ (by radix-sorting $\mathcal{A}$), and $\mathcal{C}$ (by traversing $\mathcal{B}$ and appending $i_{\max}$ to the end of $\mathcal{C}[\BWTr(i)]$ in $O(1)$ for each $i \in \mathcal{B}$
) on line~\ref{line:linear_computing-sets}  also takes $O(n)$ time. Finally, the rest of the algorithm is a scan of $\BWTr$, which can be done in $O(n)$ time, and a while loop inside that scan of $\BWTr$ which sequentially traverses in total at most the $\vert \SSS \vert$ positions in all the list in $\mathcal{C}$, so Algorithm~\ref{alg:linear} runs in linear-time. 
Correctness follows from Lemmas~\ref{lem:rm-suffixes-i_max} and \ref{lem:not-in-box}.

In other words, we have proven the following.

\begin{lemma}
Given $T\in \Sigma^n$ and $\SSS \subseteq [n]$, Algorithm~\ref{alg:linear} solves Problem~\ref{problem1} in $O(n)$ time and $O(n)$ words of space.
\end{lemma}

\subsection{Testing Minimum Suffixiency}\label{sec:cardinality}

Let $\SSS \subseteq [n]$ be a suffixient set for $T \in \Sigma^n$. By Definition~\ref{def:smallest-suffixient}, have that $\SSS$ is of the smallest cardinality for $T$ if for each $x, y \in \SSS$ such that $x \neq y$ it holds $T[x - \LCPr[bwt(x)_{\max}], x]$ does not suffix $T[y - \LCPr[bwt(y)_{\max}], y]$, and vice versa. From the above analysis, we get the following. 

\begin{proposition}\label{lem:def-smallest-suffixient}
A suffixient set $\SSS$ is a suffixient set of the smallest cardinality for $T$ if for any $x, y \in \SSS$ such that $x \neq y$, it holds $\BWTr[bwt(x)] \neq \BWTr[bwt(y)]$ or $box(bwt(x)_{\max}) \cap box(bwt(y)_{\max}) = \emptyset$.
\end{proposition}
\begin{proof}
Let $k = bwt(x)$ and $h = bwt(y)$. By Lemma~\ref{lem:rm-suffixes-i_max} we have that $P = T[x - \LCPr[k_{\max}], x]$ suffixes $Q = T[y - \LCPr[h_{\max}], y]$ if and only if $\BWTr[k] = \BWTr[h]$ and $\PSV[k] < h_{\max} < \NSV[k]$. 
Similarly, we have that $Q$ suffixes $P$ if and only if $\BWTr[k] = \BWTr[h]$ and $\PSV[h] < k_{\max} < \NSV[h]$.
Then, if neither $P$ suffixes $Q$ nor $Q$ suffixes $P$ we have $\BWTr[k] \neq \BWTr[h]$ or $h_{\max} \not\in box(k)$ and $k_{\max} \not \in box(h)$, from which we conclude $box(k_{\max}) \cap box(h_{\max}) = \emptyset$.\qed
\end{proof}


We are now ready to show a linear-time algorithm to solve Problem \ref{problem2}. This procedure is summarized in Algorithm~\ref{alg:cardinality} in Appendix~\ref{sec:algs}.

The algorithm uses the same set $\mathcal{C}$ as the one used in Algorithm~\ref{alg:linear}
, and if for some $c \in \Sigma$ and $j \in [2, \vert \mathcal{C}[c] \vert]$ it holds $\PSV[\mathcal{C}[c][j]] < \NSV[\mathcal{C}[c][j - 1]]$ (so, their boxes overlap), we know that $\SSS$ cannot be a suffixient set of the smallest cardinality. 
With respect to the running time, it takes $O(n)$ time to compute the same arrays as those used in Algorithm~\ref{alg:linear}. 
In addition, we sequentially scan the array $\mathcal{C}$, which has exactly $\vert \SSS \vert$ entries, so it takes $O(n)$ time. 
Correctness follows from Lemma~\ref{lem:rm-suffixes-i_max} and Proposition \ref{lem:def-smallest-suffixient}.

In other words, we have proven the following.

\begin{lemma}
Given $T\in \Sigma^n$ and $\SSS \subseteq [n]$, Algorithm~\ref{alg:cardinality} solves Problem~\ref{problem2} in $O(n)$ time and $O(n)$ words of space.
\end{lemma}

Since Algorithm~\ref{alg:cardinality} uses the same data structures as Algorithm~\ref{alg:linear} and it consists of a single scan of $\mathcal{C}$, we can add it as an additional step of Algorithm~\ref{alg:linear} without changing the time bounds of both algorithms. We describe them separately for ease of explanation. In Algorithm~\ref{alg:suffixiency-cardinality}, we show both tests working on a single algorithm.

\begin{algorithm}[ht]
\caption{Determines if $\SSS$ is a smallest suffixient set for a string $T\in\Sigma^n$}\label{alg:suffixiency-cardinality}
\Input{A text $T[1,n] \in \Sigma^n$ and $\SSS \subseteq [n]$}
\Output{$true$ if $\SSS$ is suffixient for $T$, $false$ otherwise}
{$\BWTr \gets \BWT(\rev{T})$;
$\LCPr \gets \LCP(\rev{T})$;
$\SAr^{-1} \gets \SA^{-1}(\rev{T})$\;}
{$\PSV \gets \PSV(\LCPr)$;
$\NSV \gets \NSV(\LCPr)$\;}
{$\mathcal{A} \gets \{bwt(s) = \SAr^{-1}[n - s + 1] \mid s \in \SSS\}$;
$\mathcal{B}[1, \vert A \vert] \gets sort(\mathcal{A})$\;}
{$\mathcal{C}[1, \sigma] \gets ([~], \dots , [~])$;
$\mathcal{P}[1, \sigma] \gets (1, \dots, 1)$\;}
\lFor{$k = 1, \dots, \vert\mathcal{B}\vert$}{$\mathcal{C}[\BWTr[\mathcal{B}[k]]].append(\mathcal{B}[k]_{\max})$}
\For(\tcp*[f]{Testing suffixiency})
{$i = 2, \dots, n$}{
\If{$\BWTr[i-1] \neq \BWTr[i]$}{
\For{$i' \in \{i - 1, i \}$}{
$c \gets \BWTr[i']$\;
\lIf{$\vert \mathcal{C}[c] \vert = 0$}{\textbf{return} $false$}
\While{$\mathcal{C}[c][\mathcal{P}[c]] \leq \PSV[i]$}{
\lIf{$\vert \mathcal{C}[c] \vert = \mathcal{P}[c]$}{\textbf{return} $false$}
$\mathcal{P}[c] \gets \mathcal{P}[c] + 1$\;
}
\lIf{$\NSV[i] \leq \mathcal{C}[c][\mathcal{P}[c]]$}{\textbf{return} $false$}
}}}
\For(\tcp*[f]{Testing minimality}){$i = 1, \dots, \sigma$}{
\For{$j = 2, \dots, \vert \mathcal{C}[i] \vert$}{
\lIf{$\PSV[\mathcal{C}[i][j]] < \NSV[\mathcal{C}[i][j - 1]]$}{\textbf{return} $false$}
}}
\Return $true$\;
\end{algorithm}

Additionally, in Figure~\ref{fig:test-example} of Appendix~\ref{sec:example}, we show an example of how Algorithms~\ref{alg:linear} and \ref{alg:cardinality} work.

\bibliographystyle{splncs04}
\bibliography{biblio}

\appendix 

\section{Example of data structures used to test Suffixiency and Minimum Suffixiency}\label{sec:example}
In Figure~\ref{fig:test-example} we show an example of the data structures used to test Suffixiency and Minimum Suffixiency. In the caption of the image it is described how the algorithms works.

\begin{figure}[!htb]
\centering
\includegraphics[scale=0.7]{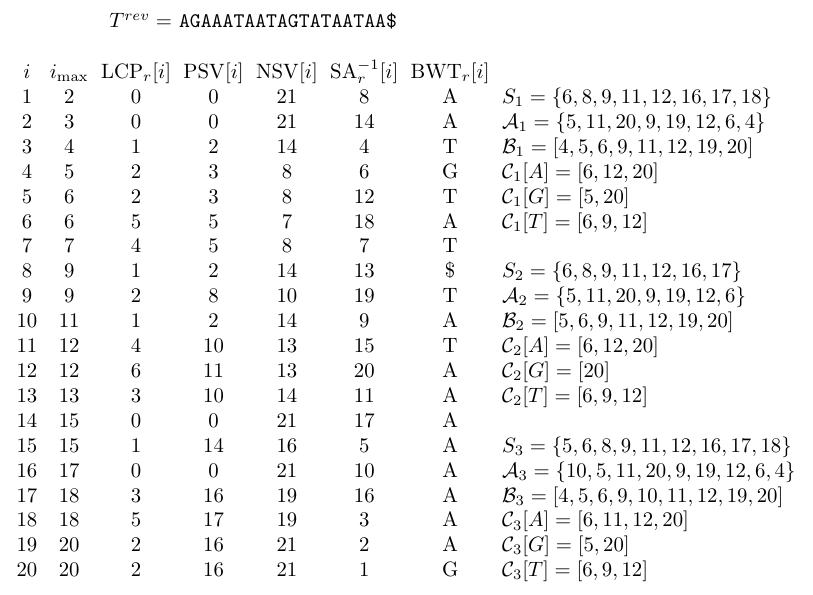}
\caption{
The figure shows the data used by Algorithm~\ref{alg:linear} 
to solve problems \ref{problem1} and \ref{problem2}. There are three sets on the right of the figure: $S_1$ (which is suffixient of smallest cardinality), $S_2$ (which is not suffixient), and $S_3$ (which is suffixient but not of smallest cardinality). They are shown with the respective sets $\mathcal{A}_j$, $\mathcal{B}_j$, and $\mathcal{C}_j$ ($j \in [1, 3]$) used in the algorithms. For brevity, we show how (i) Algorithm~\ref{alg:linear} works on the $G$-run breaks over the sets $S_1$ and $S_2$, and how (ii) Algorithm~\ref{alg:cardinality} works over the sets $S_1$ and $S_3$, and the lists $ \mathcal{C}_1[A]$ and $ \mathcal{C}_3[A]$. (i) 
For the first $G$-run break at position $4$, we have $\PSV[4] = 3 < \mathcal{C}_1[G][P[G] = 1] = 5 < \NSV[4] = 8$. For the next $G$-run break at position $5$, we have $\PSV[5] = 3 < \mathcal{C}_1[G][P[G] = 1] = 5 < \NSV[5] = 8$. The next $G$-run break occurs at position $20$. Since we have $\mathcal{C}_1[G][P[G] = 1] = 5 < \PSV[20] = 16$, we increment $\mathcal{P}[G] = \mathcal{P}[G] + 1 = 2$. Now, we have $\PSV[20] = 16 < \mathcal{C}_1[G][P[G] = 2] = 20 < \NSV[20] = 21$, and we conclude that every $G$-run break $box$ is covered by some element in $\mathcal{C}_1[G]$. In the case of set $S_2$, for the first $G$-run break at position $4$ we have $\PSV[4], \NSV[4] < \mathcal{C}_2[G][P[G] = 1] = 20$, so the algorithm rejects $S_2$ as a suffixient set. (ii) For $\mathcal{C}_1[A]$, we have $\NSV[\mathcal{C}_1[A][1] = 6] = 7 < 11 = \PSV[\mathcal{C}_1[A][2] = 12]$, $\NSV[\mathcal{C}_1[A][2] = 12] = 13 < 16 = \PSV[\mathcal{C}_1[A][3] = 20]$, so we conclude that none of the $G$-right-extension included in $\mathcal{S}_1$ suffixes another of them. In the case of set $S_3$, we have $\PSV[\mathcal{C}_3[A][3] = 12] = 11 < 13 = \NSV[\mathcal{C}[A][2] = 11]$, so the two boxes overlap and Algorithm~\ref{alg:cardinality} rejects the set $S_3$ as a suffixient set of smallest cardinality.
}\label{fig:test-example}
\end{figure}

\section{Algorithms testing Suffixiency and Minimum Suffixiency}\label{sec:algs}
In this section we show the algorithms described in Section~\ref{sec:testing-suffixiency}. In Algorithm~\ref{alg:linear} we show the algorithm solving Problem~\ref{problem1} described in Section~\ref{sec:linear}. In addition, in Algorithm~\ref{alg:cardinality} we show the algorithm solving Problem~\ref{problem2} described in Section~\ref{sec:cardinality}. 

\begin{algorithm}[tb]
\caption{Determines if $\SSS$ is a suffixient set for a string $T\in\Sigma^n$}\label{alg:linear}
\Input{A text $T[1,n] \in \Sigma^n$ and $\SSS \subseteq [n]$}
\Output{$true$ if $\SSS$ is suffixient for $T$, $false$ otherwise}
{$\BWTr \gets \BWT(\rev{T})$;
$\LCPr \gets \LCP(\rev{T})$;
$\SAr^{-1} \gets \SA^{-1}(\rev{T})$\;}\label{line:linear_computing-arrays}
{$\PSV \gets \PSV(\LCPr)$;
$\NSV \gets \NSV(\LCPr)$\;}
\label{line:linear_computing-psv-nsv}
{$\mathcal{A} \gets \{bwt(s) = \SAr^{-1}[n - s + 1] \mid s \in \SSS\}$;
$\mathcal{B}[1, \vert A \vert] \gets sort(\mathcal{A})$\;}
{$\mathcal{C}[1, \sigma] \gets ([~], \dots , [~])$;
$\mathcal{P}[1, \sigma] \gets (1, \dots, 1)$\;}\label{line:linear_computing-sets}
\lFor{$k = 1, \dots, \vert\mathcal{B}\vert$}{$\mathcal{C}[\BWTr[\mathcal{B}[k]]].append(\mathcal{B}[k]_{\max})$}
\For{$i = 2, \dots, n$}{
\If{$\BWTr[i-1] \neq \BWTr[i]$}{
\For{$i' \in \{i - 1, i \}$}{
$c \gets \BWTr[i']$\;
\lIf{$\vert \mathcal{C}[c] \vert = 0$}{\textbf{return} $false$}
\While{$\mathcal{C}[c][\mathcal{P}[c]] \leq \PSV[i]$}{
\lIf{$\vert \mathcal{C}[c] \vert = \mathcal{P}[c]$}{\textbf{return} $false$}
$\mathcal{P}[c] \gets \mathcal{P}[c] + 1$\;
}
\lIf{$\NSV[i] \leq \mathcal{C}[c][\mathcal{P}[c]]$}{\textbf{return} $false$}
}}}
\Return $true$\;
\end{algorithm}

\begin{algorithm}[!htb]
\caption{Determines if a suffixient set $\SSS$ is of the smallest cardinality}\label{alg:cardinality}
\Input{$T\in\Sigma^n$, $\PSV = \PSV(\LCPr)$, $\NSV = \NSV(\LCPr)$, a suffixient set $\SSS \subseteq [n]$, and the set $\mathcal{C}$ described in Section~\ref{sec:linear}}
\Output{$true$ if $\SSS$ is a suffixient set of smallest cardinality, $false$ otherwise}
\For{$i = 1, \dots, \sigma$}{
\For{$j = 2, \dots, \vert \mathcal{C}[i] \vert$}{
\textbf{if $\PSV[\mathcal{C}[i][j]] < \NSV[\mathcal{C}[i][j - 1]]$ then} \textbf{return} $false$\;
}}
\Return $true$\;
\end{algorithm}

\section{Deferred material from Section~\ref{sec:preliminaries}}\label{sec:deferred-preliminaries}

Given $\alpha, \beta \in \Sigma^*$, the lexicographic order $<_{lex}$ is defined as follows: $\alpha <_{lex} \beta$ if and only if there exists $j \in \mathbb{N}$ such that $\alpha[1, j] = \beta[1, j]$ and $j = \vert \alpha \vert$ or $\alpha[j + 1] \neq \beta[j + 1]$. 

Given $T\in\Sigma^n$, the suffix array $\SA(T)$ of $T$ is a permutation of $[n]$ such that for any $1 \leq i < j \leq n$ holds $T[\SA(T)[i], n] <_{lex} T[\SA(T)[j], n]$.

The longest common prefix $lcp(\alpha, \beta)$ between $\alpha$ and $\beta$ is defined as $lcp(\alpha, \beta) = \lambda$ if and only if for every $i \in [\lambda]$ it holds $\alpha[i] = \beta[i]$ and $\lambda = \min(\vert \alpha \vert, \vert \beta \vert)$ or $\alpha[\lambda + 1] \neq \beta[\lambda + 1]$.

The Longest Common Prefix array $\LCP(T)$ of $T$ is defined as $\LCP(T)[i] = lcp(T[\SA(T)[i - 1], n], T[\SA(T)[i], n])$, if $i \in [2, n]$, and $\LCP(T)[i] = 0$, otherwise.

The Burrows-Wheeler Transform $\BWT(T)$ is a permutation of the characters of $T$ such that $\BWT(T)[i]$ the character in $T$ to the left of the starting position of the $i$-th suffix of $T$, namely $\BWT(T)[i] = T[\SA[i] - 1]$, if $i \in [2, n]$, and $\BWT(T)[i] = T[n]$, otherwise.


\end{document}